\documentclass{llncs}
\usepackage{amssymb,amsmath}
\def\B{\{0,1\}}
\newcommand{\NN}{\mathbb{N}}
\newcommand{\diam}{\mathrm{diam}}
\newcommand{\ov}[1]{\overline{#1}}

\begin{document}

\title{Asynchronous simulation of Boolean networks by monotone Boolean networks}

\author{Tarek Melliti\inst{1} \and Damien Regnault\inst{1} \and Adrien Richard\inst{2}\thanks{Corresponding author} \and Sylvain Sen\'e\inst{3}}

\institute{%
Universit\'e d'\'Evry Val-d'Essonne, CNRS, IBISC EA 4526, 91000 \'Evry, France\\
\email{\{tarek.melliti,damien.regnault\}@ibisc.univ-evry.fr}
\and 
Universit\'e de Nice Sophia Antipolis, CNRS, I3S UMR 7271, 06900 Sophia-Antipolis, France\\
\email{richard@unice.fr}
\and 
Aix-Marseille Universit\'e, CNRS, LIF UMR 7279, 13288 Marseille, France\\
\email{sylvain.sene@lif.univ-mrs.fr}
}

\maketitle

\begin{abstract} 
We prove that the fully asynchronous dynamics of a Boolean network $f:\{0,1\}^n\to\{0,1\}^n$ without negative loop can be simulated, in a very specific way, by a monotone Boolean network with $2n$ components. We then use this result to prove that, for every even $n$, there exists a monotone Boolean network $f:\{0,1\}^n\to\{0,1\}^n$, an initial configuration $x$ and a fixed point $y$ of $f$ such that: (i) $y$ can be reached from $x$ with a fully asynchronous updating strategy, and (ii) all such strategies contains at least $2^{\frac{n}{2}}$ updates. This contrasts with the following known property: if $f:\{0,1\}^n\to\{0,1\}^n$ is monotone, then, for every initial configuration $x$, there exists a fixed point $y$ such that $y$ can be reached from $x$ with a fully asynchronous strategy that contains at most $n$ updates.  
\end{abstract}


\section{Introduction}

A {\em Boolean network} with $n$ components is a discrete dynamical system usually defined by a global transition function
\[
f:\B^n\to\B^n,\qquad x=(x_1,\dots,x_n)\mapsto f(x)=(f_1(x),\dots,f_n(x)).
\]
Boolean networks have many applications. In particular, since the seminal papers of McCulloch and Pitts \cite{MP43}, Hopfield \cite{H82}, Kauffman \cite{K69,K93} and Thomas \cite{T73,TA90}, they are omnipresent in the modeling of neural and gene networks (see \cite{B08,N15} for reviews). They are also essential tools in Information Theory, for the network coding problem \cite{ANLY00,GRF14}. 

\smallskip
The structure of a Boolean network $f$ is usually represented via its {\bf interaction graph}, which is the signed digraph $G(f)$ defined as follows: the vertex set is $[n]:=\{1,\dots,n\}$ and, for all $i,j\in [n]$, there exists a positive (resp. negative) arc from $j$ to $i$ is there exists $x\in\B^n$ such that 
\[
f_i(x_1,\dots,x_{j-1},1,x_{j+1},\dots,x_n)-f_i(x_1,\dots,x_{j-1},0,x_{j+1},\dots,x_n)
\]
is positive (resp. negative). Note that $G(f)$ may have both a positive and a negative arc from one vertex to another. Note also that $G(f)$ may have {\em loops}, that is, arcs from a vertex to itself. The sign of a cycle of $G(f)$ is, as usual, the product of the signs of its arcs (cycles are always directed and without ``repeated'' vertices).

\smallskip
From a dynamical point of view, there are several ways to derive a dynamics from $f$, depending on the chosen {\em updating strategy}. With the so-called {\em synchronous} or {\em parallel strategy}, each component is updated at each step: if $x^t$ is the configuration of the system at time $t$, then $f(x^t)$ is the configuration of the system at time $t+1$. Hence, the dynamics is just given by the successive iterations of $f$. On the opposite way, with the so-called {\em (fully) asynchronous strategy}, exactly one component is updated at each time. This strategy is very often used in practice, in particular in the context of gene networks \cite{TA90}. More formally, given an infinite sequence $i_0i_1i_2\dots$ of indices taken in $[n]$, the dynamics of $f$ resulting from an initial configuration $x^0$ and the asynchronous strategy $i_0i_1i_2\dots$ is given by the following recurrence: for all $t\in\NN$ and $i\in [n]$, $x^{t+1}_i=f_i(x^t)$ if $i=i_t$ and $x^{t+1}_i=x^t$ otherwise. 

\smallskip
All the possible asynchronous dynamics can be represented in a compact way by the so-called {\bf asynchronous graph $\Gamma(f)$}, defined as follows: the vertex set is $\B^n$ and, for all $x,y\in\B^n$, there is an arc from $x$ to $y$, called {\em transition}, if there exists $i\in [n]$ such that $f_i(x)=y_i\neq x_i$ and $y_j=x_j$ for all $j\neq i$. Note that $f$ and $\Gamma(f)$ share the same information. The {\em distance} between two configurations $x$ and $y$ in $\Gamma(f)$, denoted $d_{\Gamma(f)}(x,y)$, is the minimal length of a path of $\Gamma(f)$ from $x$ to $y$, with the convention that the distance is $\infty$ if no such paths exist. Note that $d_{\Gamma(f)}(x,y)$ is at least the Hamming distance $d_H(x,y)$ between $x$ and $y$. A path from $x$ to $y$ in $\Gamma(f)$ is then called  a {\bf geodesic} if its length is exactly $d_H(x,y)$. In other words, a geodesic is a path along which each component is updated at most one. The {\bf diameter} of $\Gamma(f)$ is 
\[
\diam(\Gamma(f)):=\max\{d_{\Gamma(f)}(x,y):x,y\in\B^n,d_{\Gamma(f)}(x,y)<\infty\}. 
\]

\smallskip
In many contexts, as in molecular biology, the first reliable information are represented under the form of an interaction graph, while the actual dynamics are very difficult to observe \cite{N15,TK01}. A natural question is then the following: {\em What can be said about $\Gamma(f)$ according to $G(f)$ only?} 

\smallskip
Robert proved the following partial answer \cite{R86,R95}.

\begin{theorem}\label{thm:robert}
If $G(f)$ is acyclic then $f$ has a unique fixed point $y$. Furthermore, $\Gamma(f)$ is acyclic and, for every configuration $x$, $\Gamma(f)$ has a geodesic from $x$ to~$y$.
\end{theorem}

In other words, $d_{\Gamma(f)}(x,y)=d_H(x,y)$ for every $x\in\B^n$. However, the acyclicity of $G(f)$ is not sufficient for $\Gamma(f)$ to have a short diameter. Indeed, in a rather different setting, Domshlak \cite{D02} proved (a slightly stronger version of) the following result.  

\begin{theorem}\label{thm:diam_acyclic}
For every $n\geq 8$ there exists $f:\B^n\to\B^n$ such that $G(f)$ is acyclic and $\diam(\Gamma(f))\geq 1.5^{\frac{n}{2}}$.
\end{theorem}

Now, what can be said if $G(f)$ contains cycles ? Thomas highlighted the fact that the distinction between positive and negative cycles is highly relevant (see \cite{TA90,TK01} for instance). The subtlety and versatility of the influences of interactions between positive and negative cycles lead researchers to first focus on networks with only positive cycles or only negative cycles. In particular, the following basic properties was proved in \cite{A08,RRT08,R10}: {\em If $G(f)$ has no positive (resp. negative) cycles, then $f$ has at most (resp. at least) one fixed point.} This gives a nice proof by dichotomy of the first assertion in Theorem~\ref{thm:robert}.

\smallskip
In \cite{MRRS13}, the authors showed that the absence of negative cycles essentially corresponds to the study of {\bf monotone networks}, that is, Boolean networks $f:\B^n\to\B^n$ such that 
\[
x\leq y~\Rightarrow~f(x)\leq f(y)
\]
where $\leq$ is the usual partial order ($x\leq y$ if and only if $x_i\leq y_i$ for all $i\in[n]$). More precisely, they proved the following: {\em If $G(f)$ is strongly connected and without negative cycles, then there exists a monotone network $f':\B^n\to\B^n$ such that: $G(f)$ and $G(f')$ have the same underlying unsigned digraph, and $\Gamma(f)$ and $\Gamma(f')$ are isomorphic}. Furthermore, they proved the following reachability result, that shares some similarities with Theorem~\ref{thm:robert}. 

\begin{theorem}
If $f:\B^n\to\B^n$ is monotone, then, for every configuration~$x$, $\Gamma(f)$ has a geodesic from $x$ to a fixed point $y$ of $f$. 
\end{theorem}

Here, we prove the following theorem, that shows that there may exist, under the same hypothesis, a configuration $x$ and a fixed point $y$ such that $y$ is reachable from $x$ with paths of exponential length only. This result contrasts with the previous one, and may be seen as an adaptation of Theorem~\ref{thm:diam_acyclic} for monotone networks.  

\begin{theorem}\label{thm:diam_monotone}
For every even $n$, there exists a monotone network $f:\B^n\to\B^n$, two configurations $x$ and $y$ such that $y$ is a fixed point of $f$ and 
\[
\diam(\Gamma(f))\geq d_{\Gamma(f)}(x,y)\geq  2^{\frac{n}{2}}. 
\]
\end{theorem}

The proof is by construction, and the idea for the construction is rather simple. Let $A$, $B$ and $C$ be the sets of configurations that contains $n/2-1$, $n/2$ and $n/2+1$ ones. Clearly, $A$, $B$ and $C$ are antichains of exponential size, and, in these antichains, obviously, the monotonicity of $f$ doesn't apply. This leaves enough freedom to defined $f$ on $A\cup B\cup C$ in such a way that subgraph $\Gamma(f)$ induced by $A\cup B\cup C$ contains a configuration $x$ and  fixed point $y$ reachable from $x$ with paths of exponential length only. To obtain a network as in the theorem, it is then sufficient to extend $f$ on the whole space $\B^n$ by keeping $f$ monotone and without creating shortcuts from $x$ to $y$ in the asynchronous graph. This idea, that consists in using large antichains to construct special monotone functions, is also present in \cite{GRR15} and \cite{ADG04c} for instance. 

\smallskip
Let $f:\B^n\to\B^n$ be any Boolean network such that $G(f)$ has no negative loops. With the technic described above, we can go further and prove that $\Gamma(f)$ can be embedded in the asynchronous graph $\Gamma(f')$ of a monotone network $f':\B^{2n}\to\B^{2n}$ in such a way that fixed points and distances between configurations are preserved. The formal statement follows. If $x,y\in\B^n$, then the concatenation $(x,y)$ is seen as a configuration of $\B^{2n}$ and, conversely, each configuration in $\B^{2n}$ is seen as the concatenation of two configurations in $\B^n$. As usual, we denote by $\ov{x}$ the configuration obtained from $x$ by switching every component. 

\begin{theorem}[Main results]\label{thm:embedding} 
Let $f:\B^n\to\B^n$. If $G(f)$ has no negative loops, then there exists a monotone network $f':\B^{2n}\to\B^{2n}$ such that the following two properties holds. First, $x$ is a fixed point of $f$ if and only if $(x,\ov{x})$ is a fixed point of $f'$. Second, for all $x,y\in\B^n$, $\Gamma(f)$ has a path from $x$ to $y$ of length $\ell$ if and only if $\Gamma(f')$ has a path from $(x,\bar x)$ to $(y,\bar y)$ of length~$2\ell$. 
\end{theorem}

Theorem~\ref{thm:diam_monotone} is now an easy corollary of Theorem~\ref{thm:embedding}. 

\begin{proof}[of Theorem~\ref{thm:diam_monotone} assuming Theorem~\ref{thm:embedding}]
Let $r=2^n$, and let $x^1,x^2,\dots,x^r$ be any enumeration of the elements of $\B^n$ such that $d_H(x^k,x^{k+1})=1$ for all $1\leq k<r$ (take the Gray code for instance). Let $f:\B^n\to\B^n$ be defined by $f(x^k)=x^{k+1}$ for all $1\leq k<r$ and $f(x^r)=x^r$. Let $x=x^0$ and $y=x^r$. Then $y$ is the unique fixed point of $f$. Furthermore, since the set of transitions of $\Gamma(f)$ is $\{x^k\to x^{k+1}:1\leq k<r\}$, we deduce that $d_{\Gamma(f)}(x,y)=2^n-1$. We also deduce that $G(f)$ has no negative loops (this is an easy exercise to prove that $G(f)$ has a negative loop if and only if $\Gamma(f)$ has a cycle of length two). Hence, by Theorem~\ref{thm:embedding}, there exists a monotone network $f':\B^{2n}\to\B^{2n}$ such that $(y,\ov{y})$ is a fixed point and 
\[
d_{\Gamma(f')}((x,\ov{x}),(y,\ov{y}))=2d_{\Gamma(f)}(x,y)=2^{n+1}-2\geq 2^n.
\]
\qed
\end{proof}

The paper is organized as follows. The proof of Theorem~\ref{thm:embedding} is given in Section~\ref{sec:proof}. A conclusion and some open questions are then given in Section~\ref{sec:conclusion}. 

\section{Proof of Theorem~\ref{thm:embedding}}\label{sec:proof}

We first fix some notations:
\begin{align*}
\ov{x}^i&:=(x_1,\dots,\ov{x_i},\dots,x_n)&(\textrm{$x\in\B^n$ and $i\in [n]$}),\\
w(x)&:=|\{i\in[n]:x_i=1\}|&(\textrm{$x\in\B^n$}),\\
w(x,y)&:=w(x)+w(y)&(\textrm{$x,y\in\B^n$}),\\
\Omega&:=\{(x,\ov{x}):x\in\B^n\}. 
\end{align*}

\smallskip
The function $f'$ in Theorem~\ref{thm:embedding} is defined as follows from $f$. 

\begin{definition}
Given $f:\B^n\to\B^n$, we define $f':\B^{2n}\to\B^{2n}$ by: for all $i\in [n]$ and $x,y\in\B^n$, 
\[
f'_i(x,y)=
\left\{
\begin{array}{ll}
f_i(x)&\text{ if $y=\ov{x}$ or $\ov{y}^i=\ov{x}$}\\[2mm]
\ov{x_i}&\text{ if $w(x,y)=n$ and $y\neq \ov{x}$}\\[2mm]
1&\text{ if $w(x,y)=n+1$ and $\ov{y}^i\neq\ov{x}$}\\[2mm]
0&\text{ if $w(x,y)=n-1$ and $\ov{y}^i\neq\ov{x}$}\\[2mm]
1&\text{ if $w(x,y)\geq n+2$}\\[2mm]
0&\text{ if $w(x,y)\leq n-2$}
\end{array}
\right.
\quad\text{and}\quad
f'_{n+i}(x,y)=\ov{f'_i(\ov{y},\ov{x})}.
\]
\end{definition}

\begin{remark}
$f'_i(x,y)=\ov{f'_{n+i}(\ov{y},\ov{x})}$. 
\end{remark}

\begin{remark}
Let $A$, $B$ and $C$ be sets of configurations $(x,y)\in\B^{2n}$ such that $w(x,y)$ is $n-1$, $n$ and $n+1$, respectively (these are the three sets discussed in the introduction) (we have $\Omega\subseteq B$). One can see that $f'_i$ behave as $f_i$ when $x$ and $y$ are mirroring each other ($y=\ov{x}$) or almost mirroring each other ($\ov{y}^i=\ov{x}$); and in both cases, $(x,y)$ lies in $A\cup B\cup C$. One can also see that $f'_i$ equals $0$ below the layer $A$ and equals $1$ above the layer $C$. The same remarks apply on $f'_{n+i}$, excepted that $f'_{n+1}$ behaves as the negation $\ov{f_i}$ in $A\cup B\cup C$. Hence, roughly speaking, $f$ behaves as $(f,\ov{f})$ in the middle layer $A\cup B\cup C$, and it converges toward the all-zeroes or all-ones configuration outside this layer.
\end{remark}

\begin{lemma}\label{lem:monotone} 
If $G(f)$ has no negative loops, then $f'$ is monotone.
\end{lemma}

\begin{proof}
Suppose, for a contradiction, that there exists $a,b,c,d\in\B^n$ and $i\in [n]$ such that 
\[
(a,b)<(c,d)
\text{ and }
f'_i(a,b)>f'_i(c,d).
\]
Then we have 
\[
n-1\leq w(a,b)<w(c,d)\leq n+1.
\]
This leaves three possibilities.
\begin{description}
\item[{\it Case 1: $w(a,b)=n-1$ and $w(c,d)=n+1$}.] 
Since $f'_i(a,b)=1$, we fall in the first case of the definition of $f'_i$, that is, 
\[
f'_i(a,b)=f_i(a)=1\text{ and }\ov{b}^i=\ov{a}. 
\]
Similarly
\[
f'_i(c,d)=f_i(c)=0\text{ and }\ov{d}^i=\ov{c}.
\]
Thus 
\[
(a,b)=(a,\ov{\ov{a}}^i)<(c,d)=(c,\ov{\ov{c}}^i).
\]     
So for all $j\neq i$, we have $a_j\leq c_j$ and $\ov{a_j}=(\ov{\ov{a}}^i)_j\leq (\ov{\ov{c}}^i)_j=\ov{c_j}$ thus $c_j\leq a_j$. So $a_j=c_j$ for all $j\neq i$, that is, $c\in\{a,\ov{a}^i\}$. Since $f_i(a)<f_i(c)$ we have $c=\ov{a}^i$, and since $a\leq c$ we deduce that $a_i=0$. Thus $G(f)$ has a negative arc from $i$ to $i$, a contradiction.
\medskip
\item[{\it Case 2: $w(a,b)=n-1$ and $w(c,d)=n$}.]
As in Case 1, we have 
\[
f'_i(a,b)=f_i(a)=1\text{ and }\ov{b}^i=\ov{a}.
\]
For $f'_i(c,d)$ we have two cases. Suppose first that 
\[
f'_i(c,d)=f_i(c)=0\text{ and }d=\ov{c}.
\]
Then  
\[
(a,b)=(a,\ov{\ov{a}}^i)<(c,d)=(c,\ov{c}).
\]     
So for all $j\neq i$, we have $a_j\leq c_j$ and $\ov{a_j}=(\ov{\ov{a}}^i)_j\leq \ov{c_j}$ thus $c_j\leq a_j$. So $a_j=c_j$ for all $j\neq i$, that is, $c\in\{a,\ov{a}^i\}$. Since $f_i(a)<f_i(c)$ we have $c=\ov{a}^i$, and since $a\leq c$ we deduce that $a_i=0$. Thus $G(f)$ has a negative arc from $i$ to $i$, a contradiction. The other case is 
\[
f'_i(c,d)=\ov{c_i}=0\text{ and }d\neq \ov{c}.
\]    
First, observe that for all $j\neq i$, if $c_j=0$ then $a_j=0$ thus $1=(\ov{\ov{a}}^i)_j\leq d_j$. Since $c_i=1$ we deduce that $\ov{c}\leq d$. Now, suppose that $c_j=d_j=1$ for some $j\in [n]$. Since $w(c,d)=n$, we deduce that there exists $k\neq j$ such that $c_k=d_k=0$, and this  contradicts $\ov{c}\leq d$. Thus, for all $j\in [n]$, either $d_j=0$ or $d_j>c_j$, that is, $d\leq \ov{c}$. Thus $c=\ov{d}$, a contradiction.
\medskip
\item[{\it Case 3: $w(a,b)=n$ and $w(c,d)=n+1$}.]
We obtain a contradiction as in Case~2. 
\end{description}

So we have proven that $f'_i$ is monotone for all $i\in [n]$. It remains to prove that $f'_{n+i}$ is monotone. Using the monotony of $f'_i$ for the implication we get:
\[
\begin{array}{rcl}
(a,b)\leq (c,d)
&\iff& (\ov{c},\ov{d})\leq (\ov{a},\ov{b})\\[2mm]
&\iff& (\ov{d},\ov{c})\leq (\ov{b},\ov{a})\\[2mm]
&\Longrightarrow& f'_i(\ov{d},\ov{c})\leq f'_i(\ov{b},\ov{a})\\[2mm]
&\iff& \ov{f'_i(\ov{b},\ov{a})}\leq\ov{f'_i(\ov{d},\ov{c})}\\[2mm]
&\iff& f'_{i+n}(a,b)\leq f'_{i+n}(c,d).
\end{array}
\]
\qed
\end{proof}

\begin{lemma}\label{lem:fixed} 
For all $x\in\B^n$ we have $f(x)=x$ if and only if $f'(x,\ov{x})=(x,\ov{x})$. 
\end{lemma}

\begin{proof}
By definition we have 
\[
f(x)=x\quad\iff\quad f'_i(x,\ov{x})=x_i~\forall i\in[n].
\]
So it is sufficient to prove that 
\[
f'(x,\ov{x})=(x,\ov{x})\quad\iff\quad f'_i(x,\ov{x})=x_i~\forall i\in[n].
\]
The direction $\Rightarrow$ is obvious, and $\Leftarrow$ is a consequence of the following equivalences:
\begin{align*}
f'_i(x,\ov{x})=x_i
&\iff \ov{f'_{n+i}(\ov{\ov{x}},\ov{x})}=x_i\\
&\iff f'_{n+i}(x,\ov{x})=\ov{x_i}\\
&\iff f'_{n+i}(x,\ov{x})=(x,\ov{x})_{n+i}.
\end{align*}
\qed
\end{proof}

\begin{lemma}\label{lem:01} 
For all $x,y\in\B^n$, if $\Gamma(f')$ has a path from $(x,y)$ to $\Omega$ then $n-1\leq w(x,y)\leq n+1$. 
\end{lemma}

\begin{proof}
It is sufficient to prove that,
\[
w(x,y)\leq n-2~\Rightarrow f'(x,y)=0
\quad\text{and}\quad
w(x,y)\geq n+2~\Rightarrow f'(x,y)=1.
\]
Let $i\in[n]$. If $w(x,y)\leq n-2$ (resp. $w(x,y)\geq n+2$) then $f'_i(x,y)=0$ (resp. $f'_i(x,y)=1$) by definition. Now, if $w(x,y)\leq n-2$ then $w(\bar y,\bar x)\geq n+2$ thus
\[
f'_{n+i}(x,y)=\ov{f'_i(\bar y,\bar x)}=\ov{1}=0, 
\]
and if $w(x,y)\geq n+2$ then $w(\bar y,\bar x)\leq n-2$ thus
\[
f'_{n+i}(x,y)=\ov{f'_i(\bar y,\bar x)}=\ov{0}=1. 
\]
\qed
\end{proof}

\begin{lemma}\label{lem:transition} 
If $G(f)$ has no negative loops, then, for all $x,y\in\B^n$, the following assertions are equivalent:
\begin{enumerate}
\item[{\em (1)}] $x\to y$ is a transition of $\Gamma(f)$.
\item[{\em (2)}] $(x,\ov{x})\to(y,\ov{x})\to (y,\ov{y})$ is a path of $\Gamma(f')$.
\item[{\em (3)}] $(x,\ov{x})\to(x,\ov{y})\to (y,\ov{y})$ is a path of $\Gamma(f')$.
\item[{\em (4)}] $\Gamma(f')$ has a path from $(x,\ov{x})$ to $(y,\ov{y})$ without internal vertex in $\Omega$. 
\end{enumerate}
Furthermore, the only possible paths of $\Gamma(f')$ from $(x,\ov{x})$ to $(y,\ov{y})$ without internal vertex in $\Omega$ are precisely the ones in {\em (2)} and {\em (3)}. 
\end{lemma}

\begin{proof}
Suppose that $\Gamma(f)$ has a transition $x\to y$, and let $i\in[n]$ be such that $y=\ov{x}^i$. We have $f'_i(x,\ov{x})=f_i(x)\neq x_i$ thus $\Gamma(f')$ has a transition from $(x,\ov{x})$ to $(\ov{x}^i,\ov{x})=(y,\ov{x})$. Since 
\[
f'_{n+i}(\ov{x}^i,\ov{x})=\ov{f_i(\ov{\ov{x}},\ov{\ov{x}^i})}=\ov{f_i(x,\ov{\ov{x}^i})}=\ov{f_i(x)}=x_i\neq (\ov{x}^i,\ov{x})_{n+i},  
\]
$\Gamma(f')$ has a transition from $(\ov{x}^i,\ov{x})$ to 
\[
\ov{(\ov{x}^i,\ov{x})}^{n+i}=(\ov{x}^i,\ov{\ov{x}}^i)=(y,\ov{y}).
\]
This proves the implication $(1)\Rightarrow (2)$. Now, if $\Gamma(f')$ contains the transition $(x,\ov{x})\to (y,\ov{x})$ then there exists $i\in [n]$ such that $y=\ov{x}^i$ and $y_i=f'_i(x,\ov{x})=f_i(x)$. Thus $x\to y$ is a transition of $\Gamma(f)$. So we have $(1)\iff (2)$ and we prove similarly that $(1)\iff (3)$. 

\smallskip
Since $[(2)\text{ or }(3)]\Rightarrow (4)$ is obvious, to complete the proof it is sufficient to prove that if $\Gamma(f')$ has a path $P$ from $(x,\ov{x})$ to $(y,\ov{y})$ without internal vertex in $\Omega$ then either $P=(x,\ov{x})\to(y,\ov{x})\to (y,\ov{y})$ or $P=(x,\ov{x})\to(x,\ov{y})\to (y,\ov{y})$. Let $a$ be the configuration following $(x,\ov{x})$ in $P$, and let $b$ be the configuration following $a$ in $P$. We will prove that $b=(y,\ov{y})$ and $a=(x,\ov{y})$ or $a=(y,\ov{x})$. We have $w(a)=n\pm 1$ and thus $w(b)\in\{n-2,n,n+2\}$, but if $w(b)=n\pm 2$ then we deduce from Lemma~\ref{lem:01} that $\Gamma(f')$ has no paths from $b$ to a configuration in $\Omega$, a contradiction. Thus $w(b)=n$. Let $i\in [n]$ be such that $a=(\ov{x}^i,\ov{x})$ or $a=(x,\ov{\ov{x}}^i)$. We have four cases. 
\begin{description}
\item[{\it Case 1: $a=(\ov{x}^i,\ov{x})$ and $w(a)=n-1$}.]
Since $w(a)=n-1$ we have $x_i=1$, and thus $f'_i(x,\ov{x})=f_i(x)=0$. Also $f'_i(a)=f'_i(\ov{x}^i,\ov{x})=f_i(\ov{x}^i)=0$ since otherwise $G(f)$ has a negative loop on vertex $i$. Let $1\leq j\leq 2n$ be such that $b=\ov{a}^j$. Since $w(a)<w(b)=n$, we have $a_j=0$ and  $f'_j(a)=1$. If $1\leq j\leq n$ then $j\neq i$ (since $f'_i(a)=0$) so $\ov{\ov{x}}^j\neq\ov{\ov{x}^i}$ and since $w(\ov{x}^i,\ov{x})=n-1$, we deduce from the definition of $f'$ that $f'_j(a)=f'_j(\ov{x}^i,\ov{x})=0$, a contradiction. So $n< j\leq 2n$. Let $k=j-n$. We have 
\[
f'_j(a)=f'_{n+k}(a)=f'_{n+k}(\ov{x}^i,\ov{x})=\overline{f'_k(\ov{\ov{x}},\ov{\ov{x}^i})}=\overline{f'_k(x,\ov{\ov{x}^i})}.
\]
Since $w(\ov{x}^i,\ov{x})=n-1$ we have $w(x,\ov{\ov{x}}^i)=n+1$. So if $k\neq i$ we have $\ov{\ov{\ov{x}}^i}^k\neq\ov{x}$. Thus by the definition of $f'$ we have $f'_k(x,\ov{\ov{x}^i})=1$ thus $f'_j(a)=0$, a contradiction. We deduce that $k=i$, that is, $j=n+i$. Thus $b=\ov{a}^{n+i}=(\ov{x}^i,\ov{\ov{x}}^i)\in\Omega$, and we deduce that 
\[
P=(x,\ov{x})\to (y,\ov{x})\to (y,\ov{y}).
\] 
\item[{\it Case 2: $a=(\ov{x}^i,\ov{x})$ and $w(a)=n+1$}.]
We prove with similar arguments that
\[
P=(x,\ov{x})\to (y,\ov{x})\to (y,\ov{y}).
\]
\item[{\it Case 3: $a=(x,\ov{\ov{x}}^i)$ and $w(a)=n-1$}.]
We prove with similar arguments that
\[
P=(x,\ov{x})\to (x,\ov{y})\to (y,\ov{y}).
\]
\item[{\it Case 4: $a=(x,\ov{\ov{x}}^i)$ and $w(a)=n+1$}.]
We prove with similar arguments that
\[
P=(x,\ov{x})\to (x,\ov{y})\to (y,\ov{y}).
\]
\end{description}
\qed
\end{proof}

\pagebreak

\begin{lemma}\label{lem:path} 
If $G(f)$ has no negative loops, then for all $x,y\in\B^n$, the following two assertions are equivalent:
\begin{enumerate}
\item[{\em (1)}]
$\Gamma(f)$ has a path from $x$ to $y$ of length $\ell$. 
\item[{\em (2)}]
$\Gamma(f')$ has a path from $(x,\bar x)$ to $(y,\bar y)$ of length $2\ell$. 
\end{enumerate}
\end{lemma}

\begin{proof}
According to Lemma~\ref{lem:transition}, $x^0\to x^1\to x^2\to \cdots\to x^\ell$ is a path of $\Gamma(f)$ if and only if
\[
(x^0,\ov{x^0})\to(x^0,\ov{x^1})\to  (x^1,\ov{x^1})\to(x^1,\ov{x^2})\to  (x^2,\ov{x^2})  \cdots\to (x^\ell,\ov{x^\ell})
\]
is a path of $\Gamma(f')$. This proves $(1)\Rightarrow (2)$. To prove $(2)\Rightarrow(1)$ suppose that $\Gamma(f')$ has a path $P$ from $(x,\bar x)$ to $(y,\bar y)$ of length $2\ell$. Let $(a^0,\ov{a^0}),(a^1\ov{a^1}),\dots,(a^p,\ov{a^p})$ be the configurations of $P$ that belongs to $\Omega$, given in the order (so $a^0=x$ and $a^p=y$). According to Lemma~\ref{lem:transition}, there exists $b^1,b^2,\dots,b^p$ with $b^q\in\{(a^{q-1},\ov{a^q}),(a^q,\ov{a^{q-1}})\}$ for all $1\leq q\leq p$ such that 
\[
P=(a^0,\ov{a^0})\to b^1 \to (a^1,\ov{a^1})\to b^2\to\cdots \to b^p\to (a^p,\ov{a^p}).
\]
Thus $p=\ell$, and again by Lemma~\ref{lem:transition}, $x^0\to x^1\to\cdots \to x^\ell$ is a path of $\Gamma(f)$.\qed 
\end{proof}

Theorem~\ref{thm:embedding} result from Lemmas~\ref{lem:monotone}, \ref{lem:fixed} and \ref{lem:path}. 

\section{Conclusion and open questions}\label{sec:conclusion}

In this paper we have proved that the asynchronous graph of every $n$-component Boolean network without negative loop can be embedded in the asynchronous graph of a $2n$-component {\em monotone} Boolean network,  in such a way that fixed points and distances between configurations are preserved. A consequence of this result, which was our initial goal, is that the asynchronous graph of a monotone network may have an exponential diameter. More precisely, it may exist a configuration $x$ and a fixed point $y$ reachable from $x$ such that the distance between $x$ and $y$ is at least $2^{\frac{n}{2}}$. This contrasts with the fact that for every configuration $x$ there exists a fixed point $y$ such that the distance between $x$ and $y$ is at most~$n$.  

\smallskip
These results raise several questions. Could it be possible to embed, in a similar way, a $n$-component network {\em with} negative loops into a $m$-component monotone network? Maybe this would require $m$ to be even larger than $2n$. Besides, the embedding we propose is based on the injection $x\mapsto (x,\overline{x})$ from $\B^n$ to the balanced words of length $2n$. The well-known Knuth's balanced coding scheme \cite{K86} provides a rather simple injection from $\B^n$ to the balanced words of length $n+2\log_2n$ only. Could this technique be used to decrease the number of components in the host monotone network from $2n$ to $n+2\log_2n$? Finally, it could be interesting to study the interaction graph of monotone networks with large diameter. Does it necessarily contain long cycles, or many disjoint cycles?

\paragraph{Acknowledgment} This work has been partially supported by the project PACA APEX FRI. We wish also to thank Pierre-Etienne Meunier, Maximilien Gadouleau and an anonymous reviewer for stimulating discussions and interesting remarks. 
  
\bibliographystyle{splncs_srt}
\bibliography{bib}

\end{document}